\documentclass[runningheads]{llncs}
\usepackage{xspace}
\usepackage{graphicx}

\usepackage{enumitem}
\usepackage{wrapfig}
\usepackage{subcaption}
\usepackage[T1]{fontenc}
\usepackage[utf8]{inputenc}
\usepackage{todonotes}
\usepackage{comment}
\usepackage{soul}
\usepackage{amssymb}
\usepackage{amsmath, amssymb, graphics, graphicx}
\usepackage{algorithm}
\usepackage[noend]{algpseudocode}
\usepackage{color}
\usepackage{cases}
\usepackage{xparse}
\usepackage{xargs}
\usepackage{appendix}
\usepackage{xcolor}
\usepackage{xspace}

\usepackage{hyperref}
\usepackage{thmtools}
\usepackage{thm-restate}

\usepackage{amsthm}
\theoremstyle{remark}
\newtheorem*{claimproof}{Proof of Claim}
\newtheorem{nclaim}{Claim}
\newtheorem*{observation}{Observation}

\usepackage{relsize}

\usepackage{tikz, standalone}
\usetikzlibrary{positioning, calc, shapes, fit}

\newcommand{\xcap}{\textsc{XCAP}\xspace}

\newcommand{\apc}{\textsc{Cover by Arithmetic Progressions}\xspace}

\newcommand{\ppc}{\textsc{Exact Cover by Arithmetic Progressions}\xspace}

\newcommand{\apcmod}{\textsc{Cover by Arithmetic Progressions in $\mathbb{Z}_p$}\xspace}
\newcommand{\ppcmod}{\textsc{Exact Cover by Arithmetic Progressions in $\mathbb{Z}_p$}\xspace}

\newcommand{\capp}{\textsc{CAP}\xspace}
\newcommand{\xap}{\textsc{XCAP}\xspace}

\newcommand{\jesper}[1]{\todo[color=yellow]{#1}}

\renewcommand{\mod}{\ensuremath{\mathrm{mod}}}

\begin{document}
\title{Parameterized Algorithms for Covering by Arithmetic Progressions \thanks{Supported by the project CRACKNP that has received funding from the European Research Council (ERC) under the European Union’s Horizon 2020 research and innovation programme (grant agreement No 853234).}}
\author{Ivan Bliznets\inst{1}\orcidID{0000-0003-2291-2556} \and Jesper Nedelof\inst{2}\orcidID{0000-0003-1848-0076} \and Krisztina Szil\'agyi\inst{2}\orcidID{0000-0003-3570-0528}}
\institute{University of Groningen,  the Netherlands, \email{i.bliznets@rug.nl}  \and Utrecht University, the Netherlands \email{\{j.nederlof, k.szilagyi\}@uu.nl}}
\maketitle

\begin{abstract} 
An \emph{arithmetic progression} is a sequence of integers in which the difference between any two consecutive elements is the same.
We investigate the parameterized complexity of two problems related to arithmetic progressions, called \textsc{Cover by Arithmetic Progressions} (\textsc{CAP}) and \textsc{Exact Cover by Arithmetic Progressions} (\textsc{XCAP}). In both problems, we are given a set $X$ consisting of $n$ integers along with an integer $k$, and our goal is to find $k$ arithmetic progressions whose union is $X$. In \textsc{XCAP} we additionally require the arithmetic progressions to be disjoint. Both problems were shown to be NP-complete by Heath~[IPL'90].

We present a $2^{O(k^2)} poly(n)$ time algorithm for \textsc{CAP} and a $2^{O(k^3)} poly(n)$ time algorithm for \textsc{XCAP}. We also give a fixed parameter tractable algorithm for \textsc{CAP} parameterized below some guaranteed solution size. We complement these findings by proving that \textsc{CAP} is Strongly NP-complete in the field $\mathbb{Z}_p$, if $p$ is a prime number part of the input.



\end{abstract}
\section{Introduction}
In the \textsc{Set Cover} problem one is given a universe $U$ and a set system $\mathcal{S} \subseteq 2^U$ of subsets of $U$, along with an integer $k$. The challenge is to detect whether there exist sets $S_1,\ldots,S_k \in \mathcal{S}$ such that $\bigcup_{i=1}^k S_i = U$. Unfortunately, the problem is $W[2]$-hard parameterized by $k$~\cite[Theorem 13.28]{DBLP:books/sp/CyganFKLMPPS15}, and thus we do not expect an algorithm with Fixed Parameter Tractable (FPT) runtime, i.e. runtime $f(k)|x|^{O(1)}$ for some function $f$ and input size $|x|$.

However, \textsc{Set Cover} is incredibly expressive, and it contains many well-studied parameterized problems such as $d$-\textsc{Hitting Set} (for a constant $d$), \textsc{Vertex Cover} and \textsc{Feedback Vertex Set} (see \cite{DBLP:books/sp/CyganFKLMPPS15}), albeit the last problem requires an exponential number of elements.
While all these mentioned special cases are FPT parameterized by $k$, many other special cases of \textsc{Set Cover} remain $W[1]$-hard, and the boundary between special cases that are solvable in FPT and $W[1]$-hard has been thoroughly studied already (see e.g. Table 1 in~\cite{DBLP:journals/corr/Bringmann0MN15}).
An especially famous special case is the \textsc{Point Line Cover} problem, in which one is given points $U \subseteq \mathbb{Z}^2$ and asked to cover them with at most $k$ line segments. While it is a beautiful and commonly used exercise to show this problem is FPT parameterized by $k$,\footnote{A crucial insight is that any line containing at least $k+1$ points must be in a solution.\label{fn:ci}} many slight geometric generalizations (such as generalizing it to arbitrary set systems of VC-dimension $2$~\cite{DBLP:journals/corr/Bringmann0MN15}) are already $W[1]$-hard.

In this paper we study another special case of \textsc{Set Cover}, related to \emph{Arithmetic Progressions (APs)}. Recall that an AP is a sequence of integers of the form $a,a+d,a+2d,\ldots,a+xd$, for some integers $x$, \emph{start value} $a$ and \emph{difference} $d$. We study two computational problems: \textsc{Cover by APs} (abbreviated with \capp) and \textsc{Exact Cover by APs} (abbreviated with \xcap).
In both problems we are given a set of integers $X=\{x_1, x_2, \dots, x_n\}$, and our goal is to find the smallest number of APs $s_1, s_2, \dots, s_k$ consisting of only elements in $X$ such that their union covers\footnote{We frequently interpret APs as sets by omitting the order, and "covers" can be read as "contains".} exactly the set $X$.
In the \xap problem, we additionally require that the APs do not have common elements. 
While \capp and \xap are already known to be weakly NP-complete since the 90's~\cite{heath}, the problems have been surprisingly little studied since then. 

The study of the parameterized complexity of \capp and \xap can be motivated from several perspectives.
\begin{itemize}
\item \textbf{The }\textsc{Set Cover}\textbf{ perspective: } \capp and \xap are natural special cases of \textsc{Set Cover} for which the parameterized complexity is unclear. While the problem looks somewhat similar to \textsc{Point Line Cover}, the crucial insight\textsuperscript{\ref{fn:ci}} towards showing it is FPT in $k$ does not apply since an AP can be covered with 2 other AP's. In order to understand the jump in complexity from FPT to $W[1]$-hardness of restricted \textsc{Set Cover} problems better, it is natural to wonder whether properties weaker than the one of \textsc{Point Line Cover}\textsuperscript{\ref{fn:ci}} also are sufficient for getting FPT algorithms.
\item \textbf{The practical perspective:} There is a connection between these problems and some problems that arise during the manufacture of VLSI chips~\cite{grobman1979data}. The connection implies the NP-hardness of the latter problems. Bast and Storandt~\cite{bast2013frequency,bast2014frequency} used heuristics for these problems to compress bus timetables and speed up the process of finding the shortest routes in public transportation networks. 
\item \textbf{The additive number theory perspective}: The extremal combinatorics of covers with (generalized) APs is a very well studied in the field of additive combinatorics. This study already started in the 50's with conjectures made by among others Erd\"os (see~\cite{crittenden1970any} and the references therein), and recently results in spirit of covering sets of integers with sets of low additive energy (of which APs are a canonical example) such as Freiman's Theorem and the Balog-Szemer\'edi-Gowers Theorem also found algorithmic applications~\cite{DBLP:conf/stoc/BringmannN20,DBLP:conf/stoc/ChanL15}.
\item \textbf{The "not about graphs" perspective:} Initially, applications of FPT algorithms were mostly limited to graph problems.\footnote{There has even been a series of workshops titled "Parameterized Complexity: Not-About-Graphs" (\href{http://fpt.wdfiles.com/local--files/upcoming-conferences-dagstuhl-seminars-and-workshops/Description.pdf.pdf}{link}) to extend the FPT framework to other fields.} More recently, FPT algorithms have significantly expanded the realm of their applicability. It now includes geometry, computational social choice, scheduling, constraint satisfiability, and many other application domains. However, at this stage the interaction of number theory and FPT algorithms seems to be very limited. 
\end{itemize}

\paragraph{Our Contributions.}
Our main results are FPT algorithms for \capp and \xcap:

\begin{restatable}{theorem}{thmcovering}
\label{thm:covering}
\capp admits an algorithm running in time $2^{O(k^2)} poly(n)$.
\end{restatable}
On a high level, this theorem is proved with a bounded search tree technique (similar to \textsc{Point Line Cover}): In each recursive call we branch on which AP to use. Since $k$ is the number of APs in any solution the recursion depth is at most $k$.
The difficulty however, is to narrow down the number of recursive calls made by each recursive call.
As mentioned earlier, the crucial insight\textsuperscript{\ref{fn:ci}} does not apply since an AP can be covered with 2 other AP's.
 We achieve this by relying non-trivially on a result (stated in Theorem~\ref{thm:2k}) by Crittenden and Vanden Eynden~\cite{crittenden1970any} about covering an interval of integers with APs, originally conjectured by Erd\"os.
The proof of this theorem is outlined in Section~\ref{sec:cov}.

\begin{restatable}{theorem}{thmpartitioning}
\label{thm:partitioning}
\xcap admits an algorithm running in time $2^{O(k^3)} poly(n)$. 
\end{restatable}
On a high level, the proof of this theorem follows the proof of Theorem~\ref{thm:covering}. However, to accomodate the requirement that the selected APs are disjoint we need a more refined recursion strategy. The proof of this theorem is outlined in Section~\ref{sec:part}.

We complement these algorithms with a new hardness result. Already in the 1990s, the following was written in the paper that proved weak NP-hardness of \capp and \xcap:
\definecolor{quotemark}{gray}{0.7}
\newcommand{\faquote}[3]{%
	\begin{list}{}{%
			\setlength{\leftmargin}{0.05\textwidth}
			\setlength{\rightmargin}{0.05\textwidth}
		}
		\item[]%
		\begin{picture}(0,0)(0,0)
		\put(-15,-5){\makebox(0,0){\scalebox{3}{\textcolor{quotemark}{\bfseries``}}}}
		\end{picture}\em\ignorespaces%
		#3
		\newline%
		\makebox[0pt][l]{\hspace{0.9\textwidth}%
			\begin{picture}(0,0)(0,0)
			\put(15,10){\makebox(0,0){%
					\scalebox{3}{\textcolor{quotemark}{\rm\bfseries''}}}%
			}
			\end{picture}}%
		\textsc{--- #1}
		\mbox{}\hfill\textsl{#2}
	\end{list}
}
\faquote{Heath}{{\normalfont \cite{heath}}}{
	\textit{Because the integers used in our proofs are exponentially larger than $|X|$, we have not shown our problems to be NP-complete in the strong sense [..]. Therefore, there is hope for a pseudopolynomial time algorithm for each problem.}}

While we do not directly make progress on this question, we show that two closely related problems \emph{are} strongly NP-hard.
Specifically, if $p$ is an integer, we define an AP in $\mathbb{Z}_p$ as a sequence of the form 
\[
    a,a+d\ (\mod\ p), a+2d\ (\mod\ p),\ldots, a+xd\ (\mod\ p).
\]
In the \apcmod problem one is given as input an integer $p$ and a set $X \subseteq \mathbb{Z}_p$ and asked to cover $X$ with APs in $\mathbb{Z}_p$ that are contained in $X$ that cover $X$. In \ppcmod we additionally require the APs to be disjoint. We show the following:
\begin{restatable}{theorem}{stronglyNPC}
    \apcmod  and \ppcmod are strongly NP-complete.
\end{restatable}
While this may hint at strong NP-completeness for \capp and \xcap as well, since often introducing a (big) modulus does not incur big jumps in complexity in number theoretic computational problems (confer e.g. $k$-\text{SUM}, \textsc{Partition}, etc.), we also show that our strategy cannot directly be used to prove \capp and \xcap to be strongly-NP.
Thus this still leaves the mentioned question of Heath~\cite{heath} open.
This result is proven in Section~\ref{sec:strongnp}.

Finally, we illustrate that \capp is generalized by a variant of \textsc{Set Cover} that allows an FPT algorithm for a certain below guarantee parameterization. In particular, \capp always has a solution consisting of $|X|/2$ APs that cover all sets.
\begin{restatable}{theorem}{belowguarantee}
There is an $2^{O(k)}n^{O(1)}$ time algorithm that detects if a given set $X$ of integers can be covered with at most $|X|/2-k$ APs.
\end{restatable}
This result is proved in Section~\ref{sec:belowgua}. 

\section{Preliminaries}
For integers $a,b$ we denote by $[a,b]$ a set $\{a,a+1, \dots, b\}$, for $a=1$ instead of $[1,b]$ we use a shorthand $[b]$, i.e. $[b]=\{1,2,\dots, b\}$. For integers $a,b$ we write $a\vert b$ to show that $a$ divides $b$. For integers $a_1,\dots, a_n$ we denote their greatest common divisor by $gcd(a_1,\dots, a_n)$.

An \emph{arithmetic progression} (AP) is a sequence of numbers such that the difference of two consecutive elements is the same.
While AP is a sequence, we will often identify an AP with the set of its elements. 
We say an AP \emph{stops in between} $l$ and $r$ if it largest element is in between $l$ and $r$. We say it \emph{covers} a set $A$ if all integers in $A$ occur in it. Given an AP $a, a+d, a+2d,\dots$ we call $d$ the \emph{difference}. We record the following easy observation:
\begin{observation} 
\label{obs:interesection}
    The intersection of two APs is an AP.
\end{observation}


If $X$ is a set of integers, we denote
\[
\begin{aligned}
    X^{>c} &=\{x | x \in X, x>c \}, \quad X^{\geq c} &= \{x | x \in X, x\geq c \},\\
    \quad X^{<c} &=\{x| x \in X, x<c \}, \quad X^{\leq c}&=\{x| x \in X, x \leq c \}. 
\end{aligned}
\]

Given a set $X$ and an AP $A=\{a, a+d, a+2d,\dots\}$, we denote by $A\sqcap X$ the longest prefix of $A$ that is contained in $X$. In other words, $X\sqcap A=\{a, a+d, \dots, a+\ell d\}$, where $\ell$ is the largest integer such that $a+\ell'd\in X$ for all $\ell'\in \{0,\dots, \ell\}$.

For a set of integers $X$ and integer $p$ we denote by $X_p=\langle x \bmod p| x\in X\rangle$. Here the $\langle \rangle$ symbols indicate that we build a \emph{multiset} instead of a set (so each number is replaced with its residual class mod $p$ and we do not eliminate copies).

We call an AP $s$ \emph{infinite} if there are integers $a,d$ such that $s =(a, a+d, a+2d, \dots )$. Note that under this definition, the constant AP containing only one number and difference $0$ is also infinite.

The following result by Crittenden and Vanden Eynden will be crucial for many of our results:
\begin{theorem}[\cite{crittenden1970any}] \label{thm:2k}
Any $k$ infinite APs that cover the integers $\{1,\ldots,2^k\}$ cover the whole set of positive integers.
\end{theorem}

\section{Algorithm for \textsc{Cover by Arithmetic Progressions (CAP)}}
\label{sec:cov}
Before describing the algorithm, we introduce an auxiliary lemma. This lemma will be crucially used to narrow down the number of candidates for an AP to include in the solution to at most $2^k$. 
\begin{restatable}{lemma}{lmcovering}
\label{lem:covering}
    Let $s_0$ be an AP with at least $t+1$ elements. Let $s_1, \dots, s_k$ be APs that cover the elements $s_0(0), \dots, s_0(t-1)$, but not $s_0(t)$. The APs $s_1,\dots, s_k$ may contain other elements. Suppose that each AP $s_1, \dots, s_k$ has an element larger than $s_0(t)$. Then we have $t<2^k$.
\end{restatable}
\begin{proof}
	Suppose for contradiction that $t\geq 2^k$. Let $t_i$ be the intersection of $s_i$ and $s_0$ for $i\in [k]$. Note that by Observation~\ref{obs:interesection}, $t_1,\dots, t_k$ are APs as well. For $i\in [k]$, we define $T_i=\{j\in \{0,\dots, t-1\}:\: s_0(j)\in t_i\}$. 
	
	We claim each $T_i$ is an AP:
	To see this, let $a<b<c$ be consecutive (with respect to the sorted order) elements of $T_i$ (if $|T_i|\leq 2$, it is trivially an AP). The elements $s_0(a), s_0(b), s_0(c)$ are consecutive elements of $t_i$. Thus $s_0(b)-s_0(a)=s_0(c)-s_0(b)= d_0$, where $d_0$ is the difference of $s_0$. The above equality can be written as 
	\[
	(s_0(0)+bd_0)-(s_0(0)+ad_0)=(s_0(0)+cd_0)-(s_0(0)+bd_0).
	\]
	Thus we get $b-a=c-b$, as required.
	
	Note that the APs $T_1,\dots, T_k$ cover $\{0,\ldots, t-1\}$. For $i\in [k]$, denote by $T_i'$ the infinite extension of $T_i$ (that is, $T_i'$ contains all integers whose difference with an entry of $T_i$ is a multiple of the difference of $T_i$, if $T_i$ has only one element then $T_i'$ is constant infinite AP). By Theorem~\ref{thm:2k}, $T_1',\dots, T_k'$ cover the whole set $\mathbb{N}$. In particular, $t\in T_i'$ for some $i$. By assumption, $s_i$ has an element larger than $s_0(t)$, so $s_0(t)$ is covered by $s_i$, which leads to a contradiction. 
\end{proof}

Equipped with Lemma~\ref{lem:covering} we are ready to prove our first main theorem:

\thmcovering*

\begin{proof}
    Denote the set of integers given in the input by $X$. Without loss of generality we can consider only solutions where all APs are inclusion-maximal, i.e. solutions where none of the APs can be extended by an element of $X$. In particular, given an element $a$ and difference $d$, the AP is uniquely determined: it is equal to $\{a-\ell d, \dots, a-d, a, a+d, \dots, a+r d\}$, where $\ell, r$ are largest integers such that $a-\ell'd\in X$ for all $\ell'\in [\ell]$ and $a+r'd\in X$ for all $r'\in [r]$. 

    Our algorithm consists of a recursive function $\textsc{Covering}(C, k_1,k_2)$.
    The algorithm takes as input a set $C$ of elements and assumes there are APs $s_1,\ldots,s_{k_1}$ whose union equals $C$, so the elements of $C$ are `covered' already.
    With this assumption, it will detect correctly whether there exist $k_2$ additional APs that cover all remaining elements $X \setminus C$ from the input. Thus $\textsc{Covering}(\emptyset, 0,k)$ indicates whether the instance is a Yes-instance.
    At Line~\ref{lsc} we solve the subproblem, in which set $C$ is already covered, in $2^{k^2}poly(n)$ time if $|X\setminus C|\leq k^2$. This can be easily done by writing the subproblem as an equivalent instance of \textsc{Set Cover} (with universe $U=X\setminus C$ and a set for each AP in $U$) and run the $2^{|U|}poly(n)=2^{k^2}poly(n)$ time algorithm for \textsc{Set Cover} from~\cite{DBLP:journals/siamcomp/BjorklundHK09}. 
    
    For larger instances, we consider the $k^2+1$ smallest uncovered elements $u_1,\dots, u_{k^2+1}$ in Line~\ref{lnunc}, and guess (i.e. go over all possibilities) $u_i$ and $u_j$ such that $u_i$ and $u_j$ both belong to some AP $s$ in a solution and none of the AP's $s_1,\ldots,s_{k_1}$ stops (i.e., has its last element) in between $u_i$ and $u_j$. In order to prove correctness of the algorithm, we will show later in Claim~\ref{clm:pf}
    such $u_i$ and $u_j$ exist. 
    
    Note that $u_i$ and $u_j$ are not necessarily consecutive in $s$, so we can only conclude that the difference of $s$ divides $u_j-u_i$. We use Lemma~\ref{lem:covering} to lower bound the difference of $s$ with $(u_j-u_i)/2^k$, and after we guessed the difference we recurse with the unique maximal AP with the guessed difference containing $u_i$ (and $u_j$).
    In pseudocode, the algorithm works as follows:
    \begin{algorithm}
    \caption{Algorithm for \textsc{CAP}}
    \label{alg:cap}
    \begin{algorithmic}[1]
    \State \text{\textbf{Algorithm} $\textsc{Covering}(C,k_1,k_2)$}
    \State Let $k=k_1+k_2$
    \If{$|X \setminus C|\leq k^2$}
        \State Use the algorithm for \textsc{Set Cover} from~\cite{DBLP:journals/siamcomp/BjorklundHK09}\label{lsc}
    \Else
        \State Let $u_1,\dots, u_{k^2}$ be the $k^2+1$ smallest elements of $X\setminus C$  \label{lnunc}
        \For{$i=1\dots k^2$}
            \For{$j=i+1\dots k^2+1$}
                \State Let $D = u_j-u_i$
                \For{$\ell=1\dots 2^k$}
                    \If{$\ell$ divides $D$}
                        \State Let $s = \textsc{MakeAP}(u_i, D/\ell)$\label{ls}
                        \If{$\textsc{Covering}(C\cup s, k_1+1,k_2-1)$}
                        \State  \textbf{return true}
                        \EndIf
                    \EndIf
                \EndFor
            \EndFor
        \EndFor
        \State \textbf{return false}
    \EndIf
    \end{algorithmic}
    \end{algorithm}
 
    The procedure $\textsc{MakeAP}(a, d)$ returns the AP $\{a-\ell_1d, \dots, a-d, a, a+d, \dots, a+\ell_2d \}$, where $\ell_1$ (respectively, $\ell_2$) are the largest numbers such that $a-\ell'd\in X$ for all $\ell'\leq \ell_1$ (respectively, $a+\ell'd\in X$ for all $\ell'\leq \ell_2$).
    It is easy to see that if Algorithm \ref{alg:cap} outputs \textbf{true}, we indeed have a covering of size $k$: Since we only add elements to $C$ if they are indeed covered by an AP, and each time we add elements to $C$ because of an AP we decrease our budget $k_2$. 

    For the other direction of correctness, we first claim that our algorithm will indeed at some moment consider an AP of the solution.

\begin{nclaim}\label{clm:pf}
     Suppose there exist an AP-covering $s_1,\dots, s_k$, and let $k_1 \in \{0,\ldots,k\}$ be an integer and let $\cup_{i=1}^{k_1} s_i=C$. Then in the recursive call \textsc{Covering}$(C, k_1,k-k_1)$ we have in some iteration of the for-loops $s=s_h$ at Line~\ref{ls}, for some $h\in\{k_1+1,\dots, k\}$.
\end{nclaim}
\begin{claimproof}
    Consider the set $B$ consisting of the $k^2+1$ smallest elements of $X\setminus C$. By the pigeonhole principle, there is an $h\in \{j+1,\dots, k\}$ such that $s_h$ covers at least $k+1$ elements of $B$. Thus, there are two consecutive elements of $B\cap s_h$, $s_h(\alpha)$ and $s_h(\beta)$, between which no AP $s_1,\ldots,s_{k_1}$ ends. Without loss of generality, we may assume that $s_h(\alpha)$ and $s_h(\beta)$ are the closest such pair (i.e. the pair such that $|s_h(\alpha)-s_h(\beta)|$ is minimal). 
    By applying Lemma~\ref{lem:covering}, we conclude that there are at most $2^k-1$ elements of $s_h$ between $s_h(\alpha)$ and $s_h(\beta)$. In other words, the difference of $s_h$ is a divisor of $s_h(\beta)-s_h(\alpha)$ and at least $(s_h(\beta)-s_h(\alpha))/2^k$. Therefore, in some iteration of the for-loops we get $s=s_h$.\qed
\end{claimproof}
    Using the above claim, it directly follows by induction on $k_2=0,\ldots,k$ that if there is a covering $s_1,\ldots,s_k$ of $X$ such that $\cup_{i=1}^{k_1}s_i=C$, then the function $\textsc{Covering}(C, k-k_2,k_2)$ returns true. 

    Let us now analyse the running time of the above algorithm.
    The recursion tree has height $k$ (since we reduce $k$ by one on every level). The maximum degree is $2^k k^4$, so the total number of nodes is $2^{O(k^2)}$. The running time at each node is at most $2^{O(k^2)}poly(n)$. Therefore, the overall running time is $2^{O(k^2)}poly(n)$. 
\end{proof}
\section{FPT Algorithm for Exact Cover by Arithmetic Progressions}
\label{sec:part}
Now we show that \xap is Fixed Parameter Tractable as well. Note that this problem is quite different in character than \capp. For example, as opposed to \capp, in \xap we cannot describe an AP with only one element and its difference. Namely, it is not always optimal to take the longest possible AP: for example, if $X=\{0, 4, 6, 7, 8, 9\}$, the optimal solution uses the progression $0,4$ rather than $0,4,8$. While we still apply a recursive algorithm, we significantly need to modify our structure lemma (shown below in Lemma~\ref{lm:exactmain}) and the recursive strategy.

Before we proceed with presenting the FPT algorithm for \ppc (\textsc{XCAP}) we state an auxiliary lemma.

\begin{restatable}{lemma}{lmpart}
\label{lm:exactmain} Let $s_1, s_2, \dots, s_k$ be a solution of an instance of \xap with input set $X$.  Let $s$ be an AP that is contained in $X$. For each $i\in [k]$, we denote by $t_i$ the intersection of $s$ and $s_i$. 
Suppose that for some $i$, $t_i$ has at least $k+1$ elements. Then there are at most $2^{k-1}-1$ elements of $s$ between any two consecutive elements of $t_i$.
\end{restatable}
\begin{proof} Without loss of generality, we may assume that $i=1$. Note that $t_1$ is a subset of $s$ and AP. Hence between any two consecutive members of $t_1$ we have the same number of elements of $s$.
	Let us consider the first $k+1$ members of $t_1$, denoted by $t_1(0), t_1(1), \dots, t_1(k)$.
	For any $j\in [0, k-1]$ all elements of $s$ between $t_1(j)$ and $t_1(j+1)$ must be covered  by $s_2, s_3, \dots, s_k$. Therefore, by Lemma~\ref{lem:covering} one of the arithmetic progressions $s_2, s_3, \dots, s_k$ must stop before $t_1(1)$. Similarly, another AP from $s_2, s_3, \dots, s_k$ must stop before $t_1(2)$. By repeatedly applying this argument, we conclude that all APs $s_2, s_3, \dots, s_k$ must stop before $t_1(k-1)$. Hence, the elements of $s$ between $t_1(k-1)$ and $t_1(k)$ are uncovered which leads to a contradiction.
\end{proof}

Now we have all ingredients to prove the main theorem of this section.

\thmpartitioning*
\begin{proof}
Let $X=\{a_1,\dots, a_n\}$ be the input set. Without loss of generality, we may assume that $a_1<\dots<a_n$.
Assume that the input instance has a solution with $k$ APs: $o_1,\dots, o_k$. Let $d'_i$ be the difference of $o_i$. 
    Our algorithm $\textsc{partition}$ recursively calls itself and is described in Algorithm~\ref{alg:part}. The algorithm has the following parameters: $T_1,\dots, T_k, P_1,\dots, P_k, d_1, d_2, \dots, d_k$. The sets $T_i$ describe the elements that are in $o_i$ (i.e. the elements that are definitely covered by $o_i$). The integer $d_i$ is either $0$ or equal to the guessed value of the difference of the AP $o_i$. Once we assign a non-zero value to $d_i$, we never change it within future recursive calls. We denote by $P_i$ the set of "potentially covered" elements. Informally, $P_i$ consists of elements that could be covered by $o_i$ unless $o_i$ is interrupted by another AP. Formally, if for an AP $o_i$ we know two elements $a,b\in o_i$ ($a<b$), and the difference $d_i$ then $P_i =\{b+d_i,b+2d_i, \dots b+\ell d_i\}$ where $\ell$ is the largest number such that: $\{b+d_i,b+2d_i, \dots b+\ell d_i\}\subset X$ and none of the elements of the set $\{b+d_i,b+2d_i, \dots b+\ell d_i\}$ belong to $T_j$ for some $j\neq i$, $j\in [k]$.

\begin{algorithm}[h]
    \caption{Algorithm for \textsc{XCAP}}
    \label{alg:part}
    \begin{algorithmic}[1]
    \State \text{\textbf{Algorithm} $\textsc{Partition}(X,T_1,\dots T_k,P_1,\dots, P_k, d_1, \dots, d_k)$}
    \If{there are $i,j \in [k]$ s.t. $i \neq j$ and $P_i\cap P_j\neq \emptyset$}
    \State Let $c$ be the smallest element s.t. $\exists i,j$, $c\in P_i \cap P_j$, and $i<j$
    \State $\textsc{Partition}(X,T_1,\dots, T_k, P_1,\dots, P_{i-1}, P_i^{<c}, P_{i+1}, \dots, P_k,d_1, \dots, d_k)$
    \State $\textsc{Partition}(X,T_1,\dots, T_k, P_1, \dots, P_{j-1}, P_j^{<c}, P_{j+1},\dots, P_k,d_1, \dots, d_k)$
    \ElsIf{there is an $i$ such that $T_i=\{a_\alpha,a_\beta \}$ and $d_i=0$}
        \For{$(b_1, \dots, b_k) \in \{0,1,\dots, 2^k+1\}^k$}
            \State $g\gets gcd(a_\beta-a_\alpha, b_1d_1,\dots b_kd_k)$
            \State $D_i\gets\{g, \frac{g}{2}, \dots, \frac{g}{k(k+1)}\}$
            \For{$d \in D_i$ and $d$ is integer}
                \State $d_i \gets d$
                \State $T_i\gets \{a_\alpha, a_\alpha+d_i, \dots, a_\beta\}$
                \State $C^{\infty} \gets \{a_\beta+d_i,a_\beta+2d_i, \dots\}$
                \State $P_i\gets C^{\infty} \sqcap X$ 
                \If{ for all $j\in [k]\setminus \{i\}$ we have $T_i\cap T_j = \emptyset$} 
                \For{$j\in [k] \setminus \{i\}$}
                    \State $P_j\gets \textsc{Update}(P_j,T_i)$ \Comment{remove ints larger than min($P_j\cap T_i$)}
                \EndFor
                \State $\textsc{Partition}(X,T_1,\dots,T_k,P_1,\dots, P_k,d_1, \dots, d_k)$
                \EndIf
            \EndFor
        \EndFor
    \ElsIf{$X \setminus( P\cup T)=\emptyset$}
        \State \textbf{return} $o_1=T_1\cup P_1,\dots, o_k=T_k\cup P_k$
    \Else
        \State $a_\beta\gets min(X\setminus (T\cup P))$ 
        \If{there is an $i$ such that $|T_i|<2$}
            \State $J\gets\{i| i\in [k] \text{ and } |T_i|<2\}$
            \For{$i \in J$}
               \State $T_i\gets T_i\cup \{a_\beta\}$
               \State $\textsc{Partition}(X,T_1, \dots, T_k, P_1, \dots, P_k,d_1,\dots, d_k)$        
           \EndFor
        \Else 
            \State \textbf{abort} \Comment{this branch does not generate a  solution}
        \EndIf
    \EndIf 
    \end{algorithmic}
    \end{algorithm}
   
Initially, we call the algorithm $\textsc{partition}$ with parameters $T_1=\{a_1\}$, $T_2=\dots=T_k=P_1=\dots=P_k=\emptyset, d_1=d_2=\dots=d_k=0$.
    We denote by $T=\cup_i T_i$ and $P=\cup_i P_i$. Let $a_\beta$ be the smallest element of the input sequence that does not belong to $T\cup P$. For each $i\in [k]$ such that $|T_i|\leq 1$, we recursively call $\textsc{partition}(T_1, \dots, T_{i-1}, T_i', T_{i+1}, \dots, T_k, P_1, \dots, P_k)$, where $T_i'=T_i\cup\{a_\beta\}$. In other words, we consider all possible APs that cover $a_\beta$. We do not assign $a_\beta$ to $i$-th AP with $|T_i|\geq 2$ as for such APs we know that $a_\beta\not \in o_i$ (since $a_\beta\not\in P_i$).
 
    If in the input for some $i$ we have $|T_i|=2$ and $d_i=0$ then we branch on the value of the difference of the $i$-th AP. Let $T_i=\{a_\alpha, a_\beta\}$, where $a_\beta>a_\alpha$. By construction (specifically, by the choice of $a_
    \beta$), all elements of the sequence $B=\{a_{\alpha+1}, a_{\alpha+2},\dots, a_{\beta-1}\}$ belong to $P\cup T$, i.e. can be covered by at most $k$ APs. Note that purely based on the knowledge of two elements $a_\alpha, a_\beta \in o_i$ we cannot determine immediately the difference of $o_i$, because $a_{\alpha}, a_{\beta}$ might not be consecutive elements of $o_i$. In other words, it could happen that some elements of $B$ belong to $o_i$. Therefore we need to consider cases when potentially covered elements from some $P_j$ actually belong to $o_i$ instead of $o_j$. We note that the number of elements between $a_\alpha$ and $a_\beta$ can be very large, so a simple consideration of all cases where the difference of $o_i$ is a divisor of $a_\beta-a_\alpha$ will not provide an FPT-algorithm. 
    
    Instead of considering all divisors we use Lemma~\ref{lm:exactmain} and bound the number of candidates for value of the difference of $o_i$. 
    %
    %
    Note that we allow integers to be $0$ here, and define that all integers are a divisor of $0$. Hence $a|0$ is always true and $gcd(0,y_1,y_2,\dots)=gcd(y_1,y_2,\dots)$.
    
    Informally, for all known differences $d_j$ up to this point (i.e. all $d_j\neq 0$) we branch on the smallest positive value $b_j$ such that $d'_i \vert b_j d_j$, where $d'_i$ is the difference of the $i$-th AP in the solution that we want to find. We treat all cases when $b_j>2^k+1$ at once, and instead of the actual value we assign $0$ to a variable responsible for storing value of $b_j$. 
    Intuitively, $b_j$ describes the number of elements of $o_j$ between two consecutive interruptions by $o_i$. By Lemma~\ref{lm:exactmain}, if $b_j$ is large (larger than $2^k+1$), each of these interruptions implies that an AP stops. 
    
    Formally, for each $k$-tuple $(b_1,\dots, b_k)\in \{0,\dots, 2^k+1\}^k$, we do the following. Let $g=gcd(a_\beta-a_\alpha, b_1d_1,\dots, b_kd_k)$. From the definition of $b_j$ it follows that $d'_i\vert g$.
%
 %
%
    \begin{nclaim}\label{claim:divisors}
     If all previous branchings are consistent with a solution $o_1, \dots, o_k$ then $d'_i\geq \frac{g}{k(k+1)}.$ 
    \end{nclaim}
    \begin{claimproof}
    Indeed, if $d'_i=g/t$ and $t>k(k+1)$ then between $a_\alpha$ and $a_\alpha+g$ there are at least $k(k+1)$ elements from $o_i$. All these elements are covered by sets $P_1, \dots , P_{i-1}, P_{i+1}, \dots, P_k$. Therefore, by the pigeonhole principle there is an index $q$ such that $P_q$ contains at least $k+1$ elements. This means that $P_q$ and $o_i$ have at least $k+1$ common elements ($P_q$ is intersected by $o_i$ at least $k+1$ times). Denote these common elements by $c_1,\dots, c_t$ and let $e$ be the number of elements of $P_q$ between $c_\ell$ and $c_{\ell+1}$ ($e$ does not depend on $\ell$ as $P_q$ and $o_i$ are APs). First of all, recall that $c_{\ell}, c_{\ell+1}$ are from interval $(a_{\alpha},a_{\alpha}+g)$. Therefore, $(e+1)d_q=c_{\ell+1}-c_{\ell}<g$. Moreover, by Lemma~\ref{lm:exactmain} (applied with $s=P_q$ and $s_i=o_i$), we have $e\leq 2^{k-1}-1$. Taking into account that
    $d'_i\vert (e+1)d_q$ and $e\leq 2^{k-1}-1$ we have that $g$ divides $(e+1)d_q$ which contradicts $(e+1)d_q<g$.\qed 
    \end{claimproof}
    
    From the previous claim it follows that $d'_i \in \{g, g/2, \dots, g/k(k+1)\}$ and we have the desired bound on the number of candidates for the value of the difference of $o_i$. We branch on the value of $d'_i$, i.e. in each branch we assign to $d_i$ some integer value from the set  $\{g, g/2, \dots, g/k(k+1)\}$. In other words, for each $d\in \{g, g/2, \dots, g/k(k+1)\}$ we call 
    \[
    \textsc{partition}(T_1,\dots, T_i', \dots, T_k, P_1', \dots, P_k',d_1,\dots, d_{i-1},d,d_{i+1},\dots, d_k),
    \] where $C_d=\{a_{\alpha}, a_{\alpha}+d, \dots, a_{\beta}-d, a_{\beta}\}$, $C_d^{\infty}=\{a_{\beta}+t, a_{\beta}+2t, \dots \}$, $T_i'=T_i\cup C_d$, $P'_i=C_d^{\infty}\sqcap X$ 
    and for $j\neq i$ we set $P_j'=\textsc{Update}(P_j,T_i)$.
    The function $\textsc{Update}(A, B)$ returns $A^{<x}$, where $x=\min(A\cap B)$ (if $A\cap B=\emptyset$, it returns $A$).
    If in some branch the sequence $a_{\alpha}, a_{\alpha}+d_i, \dots, a_{\beta}-d_i, a_{\beta}$ intersects $T$, we abort this branch. Overall, in order to determine the difference of $o_i$ after discovering $a_\alpha, a_\beta \in o_i$ we create at most  $(2^k+2)^k\cdot k(k+1)=2^{O(k^2)}$ branches.

    Let us now compute the number of nodes in the recursion tree. Observe that we never remove elements from $T$.  
    Consider a path from the root of the tree to a leaf. It contains at most $2k$ nodes of degree at most $k$ (adding two first elements to $T_i$), at most $k^2$ nodes of degree 2 (resolving the intersections of two APs, lines 2-5 in pseudocode) and at most $k$ nodes of degree $2^{O(k^2)}$ (determining the difference of an AP that contains two elements). Hence, the tree has at most $k^{2k} \cdot 2^{k^2}\cdot  (2^{O(k^2)})^k = 2^{O(k^3)}$ leaves. Therefore, the overall number of nodes in the recursion tree and the running time of the algorithm is $2^{O(k^3)} poly(n)$.

    Claim~\ref{claim:divisors} proves that we iterate over all possibilities. Hence, our algorithm is correct.
\end{proof}
\section{Strong NP-hardness of Cover by Arithmetic Progressions in $\mathbb{Z}_p$}
\label{sec:strongnp}
A natural question to ask, in order to prove Strong NP-hardness for \capp, is whether we can replace an input set $X$ with an equivalent set which has smaller elements. Specifically, could we replace the input number with numbers polynomial in $|X|$, while preserving the set of APs?
This intuition can be further supported by result on Simultaneous Diophantine approximation that exactly achieve results in this spirit (though with more general properties and larger upper bounds)~\cite{10.1007/BF02579200}.
However, it turns out that this is not always the case, as we show in Lemma~\ref{lm:mod}. This means that one of the natural approaches for proving strong NP-hardness of \capp does not work: namely, not all sets $X$ can be replaced with set $X'$ of polynomial size which preserves all APs in $X$. 


\begin{restatable}{lemma}{lmmodp}
\label{lm:mod} 
Let $x_1=0$, $x_i=2^{i-2}$ for $i\geq 2$ and $X_n=\{x_1,\dots, x_{n+2}\}$. Then for any polynomial $p$ there exists an integer $n$ such that there is no set $A_n=\{a_1, \dots, a_{n+2}\}$ that satisfies the following criteria:
\begin{itemize}
    \item $a_{i}\leq p(n)$ for all $i\in[n+2]$,
    \item For all $i,j,k\in [n+2]$, the set $\{x_i,x_j,x_k\}$ forms an AP if and only if $\{a_i,a_j,a_k\}$ forms an AP.
\end{itemize}
\end{restatable}
\begin{proof}
	Assume for contradiction that there is a polynomial $p$ such that for any $n$, we can construct a set $A_n$ with elements smaller than $p(n)$ which preserves the APs of size 3 in $X_n$. Without loss of generality, we may assume that $a_1=0$ (by subtracting $a_1$ from all elements of $A_n$ and using $2p$ instead of $p$).
	
	Clearly if $|X_n|\geq 4$ we must have $a_2\neq 0$. If a set $\{0,a,b\}$ generates an AP, then $b\in \{2a, -a, \frac{a}{2}\}$. We know that $x_1,x_j,x_{j+1}$ generate an AP for each $j\in [2,n-1]$, which implies that for $i\geq 3$, $a_i$ is equal to $a_{i-1}$ multiplied by $2,-1$ or $\frac{1}{2}$.
	Note that we cannot have $q_1< q_2$ such that $a_{q_1}=a_{q_2}$. Indeed, if $q_2>q_1+1$ then $a_1, a_{q_1}, a_{q_2}$ generates an AP while $x_1, x_{q_1}, x_{q_2}$ does not. It is left to consider case $q_2=q_1+1$. If $q_2< n$, $X_n$ contains the AP $x_1,x_{q_2},x_{q_2+1}$, so $\{a_1,a_{q_2},a_{q_2+1}\}$ is an AP. Since  $a_{q_2}=a_{q_1}$, we have that $\{a_1,a_{q_1},a_{q_2+1}\}$ also generates an AP which contradicts the fact that $x_1, x_{q_1}, x_{q_2+1}$ is not an AP. If $q_2 = n$, then instead of the triple $\{a_1,a_{q_2},a_{q_2+1}\}$ we consider the triple $\{a_1,a_{q_1-1},a_{q_1}\}$ and instead of $x_1, x_{q_1}, x_{q_2+1}$ we consider $x_1, x_{q_1-1}, x_{q_2}$ and get the contradiction as $n\geq 4$. Therefore all numbers $a_1, a_2, \dots, a_n$ must be different.
	
	Recall that for each $j>1$ we have either $a_j=2^{p_j}a_2$ or $a_j=-2^{p_j}a_2$ for some integer $p_j$.  As all numbers in the set $\{a_1, a_2, \dots, a_n\}$ are different, we have that there is an index $j$ such that $p_j>\frac{n}{8}$ or $p_j<-\frac{n}{8}$. Since all numbers in $A$ must be integers we conclude that either $a_2>2^{\frac{n}{8}}$ or $a_{p_j}>2^{\frac{n}{8}}$. Therefore we must have $p(n)\geq 2^{\frac{n}{8}}$, which is not true if we take $n$ sufficiently large.  
\end{proof}
    
Unfortunately, we do not know how to directly circumvent this issue and improve the weak NP-hardness proof of Heath~\cite{heath} to \emph{strong} NP-hardness. Instead, we work with a small variant of the problem in which we work in $\mathbb{Z}_p$. The definition of APs naturally carries over to $\mathbb{Z}_p$. It is easy to see that APs are preserved:
\begin{nclaim}\label{claim:Zp}
    Let $p$ be a prime and let $X$ be a set of integers that forms an AP. Then the multiset $X_p$ generates AP in the field $\mathbb{Z}_p$.
\end{nclaim}
\begin{claimproof}
    Note that $y-x=z-y$ implies $(y \bmod p) - (x\bmod p) \equiv_p (z \bmod p) - (y\bmod p)$. Therefore, we have that $X_p$ is an AP in $\mathbb{Z}_p$. 
\end{claimproof}

However, the converse does not hold. Formally, if for some $p$ the multiset $X_p$ is an AP in $\mathbb{Z}_p$ it is not necessarily true that $X$ is an AP in $\mathbb{Z}$. For example, consider $X=\{3,6,18\}$ and $p=3$: we have $X_p=\{0,0,0\}$ which is a (trivial) AP, while $X$ is not an AP. 

We now show strong NP-completeness for the modular variants of \capp and \xcap.
In the \apcmod problem one is given as input an integer $p$ and a set $X \subseteq \mathbb{Z}_p$ and asked to cover $X$ with APs in $\mathbb{Z}_p$ that are contained in $X$ that cover $X$. In \ppcmod we additionally require the APS to be disjoint.

\stronglyNPC*

\begin{proof}
    We recall that Heath~\cite{heath} showed that \capp and \xcap are weakly NP-complete via reduction from \textsc{Set Cover}. Moreover, the instances of \capp and \xcap, obtained after reduction from \textsc{Set Cover}, consist of numbers that are bounded by $2^{q(n)}$ for some polynomial $q(n)$. To show that  \apcmod  and \ppcmod are strongly NP-complete we take a prime $p$ and convert instances of \capp and \xcap with set $S$ into instances of \apcmod and \ppcmod respectively  with a set $S_p$ (we can guarantee that the multiset $S_p$ contains no equal numbers) and modulo $p$. 

    As shown in Claim~\ref{claim:Zp}, under such transformation a \textsc{Yes}-instance is converted into a \textsc{Yes}-instance. However, if we take an arbitrary $p$ then a \textsc{No}-instance can be mapped to a \textsc{Yes}-instance or $S_p$ can become a multiset instead of a set. In order to prevent this, we carefully pick the value of $p$. 

    We need to guarantee that if $X$ is not an AP then $X_p$ also does not generate an AP in $\mathbb{Z}_p$. Suppose $X_p=\{y_1, y_2, \dots, y_k\}$ generates an AP in $\mathbb{Z}_p$ exactly in this order. We assume that $x_i$ maps into $y_i$, i.e. $x_i \equiv_p y_i $. Since $X_p$ is an AP in $\mathbb{Z}_p$, we have
    $$y_2-y_1 \equiv_p  y_3 - y_2 \equiv_p \dots  \equiv_p y_k-y_{k-1}. $$ Since $X$ is not an AP there exists an index $j$ such that
    $x_j-x_{j-1}\neq x_{j+1}-x_j$. Therefore, we have that $2x_j - x_{j-1}-x_{j+1}\neq 0$ and $2y_j - y_{j-1}-y_{j+1} \equiv_p 0$. Since $x_i \equiv_p y_i$ we conclude that $p$ divides $2x_j - x_{j-1}-x_{j+1}$. Hence if we want to choose $p$ that does not transform a \textsc{No}-instance into a \textsc{Yes}-instance, it is enough to choose $p$ such that $p$ is not a divisor of $2x-y-z$ where $x,y,z$ are any numbers from the input and   $2x-y-z\neq 0$. Similarly, if we want $S_p$ to be a set instead of a multiset, then for any different $x,y$ the prime $p$ should not be a divisor of $x-y$. Note that the number of different non-zero values of  $2x-y-z$ and $x-y$ is at most $O(n^3)$. Since all numbers are bounded by $2^{q(n)}$ the values of $2x-y-z\neq 0$ and $x-y$ are bounded by $2^{q(n)+2}$. Note that any integer $N$ has at most $\log N$ different prime divisors. Therefore at most $O(n^3) (q(n)+2)$ prime numbers are not suitable for our reduction. In order to find a suitable prime number we do the following:
    \begin{itemize}
        \item for each number from $2$ to $n^6 (q(n)+2)^2$ check if it is prime (it can be done in polynomial time~\cite{primes}), 
        \item for each prime number $p'\leq n^6(q(n)+2)^2$ check if there are integers $x,y,z$ from the input such that ($2x-y-z\neq 0$ and $2x-y-z\equiv_p 0$) or $x\equiv_p y$ if such $x,y,z$ exist go to the next prime number,
        \item when the desired prime $p'$ is found output instance $S_{p'}$ with modulo $p'$. 
    \end{itemize}

Note that number of primes not exceeding $N$ is at least $\frac{N}{2\log N}$ for large enough $N$. Hence by pigeonhole principle we must find the desired $p'$ as we consider all numbers smaller than $n^6 (q(n)+2)^2$ and $\frac{n^6 (q(n)+2)^2} { \log({n^6 (q(n)+2)^2})} > O(n^3) (q(n)+2)$ for sufficiently large $n$.

Therefore, in polynomial time we can find $p'$ that is polynomially bounded by $n$ and \apcmod with input $(S_{p'}, p')$ is equivalent to \capp  
with input $S$ (similarly for \ppcmod  and  \xcap). Hence, \apcmod and \ppcmod are strongly NP-complete.
    
    \end{proof}

\section{Parameterization Below Guarantee}
\label{sec:belowgua}
In this section we present an FPT algorithm parameterized below guarantee for a problem that generalizes \capp, namely \textsc{$t$-Uniform Set Cover}. 
The \textsc{$t$-Uniform Set Cover} is a special case of the \textsc{Set Cover} problem in which all instances $\mathcal{S},U$  satisfy the property that $\{ A \subseteq U : |A| = t \ \} \subseteq \mathcal{S}$. 
Clearly the solution for the \textsc{$t$-Uniform Set Cover} problem is at most $\lceil\frac{n}{t}\rceil$ where $n$ is the size of universe. Note that \capp is a special case of \textsc{$2$-Uniform Set Cover} since any pair of element forms an AP.
Thus we focus on presenting a fixed parameter tractable algorithm for the \textsc{$t$-Uniform Set Cover} problem parameterized below $\lceil\frac{n}{t}\rceil$ (we consider $t$ to be a fixed constant). We note that for a special case with $t=1$ the problem was considered in works~\cite{dualsetcover,crowston2013parameterized}. 

We will use the deterministic version of color-coding, which uses the following standard tools:

\begin{definition}
For integers $n,k$ a \emph{$(n,k)$-perfect hash family} is a family $\mathcal{F}$ of functions from $[n]$ to $[k]$ such that for each set $S \subseteq [n]$ of size $k$ there exists a function $f \in \mathcal{F}$ such that $\{f(v) : v \in S\}=[k]$.
\end{definition}

\begin{lemma}[\cite{DBLP:conf/focs/NaorSS95}]\label{lem:ph}
For any $n,k\geq 1$, one can construct an $(n,k)$-perfect hash family of size $e^{k}k^{O(\log k)}\log n$ in time $e^{k}k^{O(\log k)}n\log n$.
\end{lemma}

Now we are ready to state and prove the main result of this section:
\begin{theorem}\label{thm:uniformset}
    There is an $2^{O(k)} poly(n)$-time algorithm that for constant $t$ and a given instance of \textsc{$t$-Uniform Set Cover} determines the existence of a set cover of size at most $\lceil\frac{n}{t}\rceil-k$ where $k$ is an integer parameter and $n$ is the size of the universe. 
\end{theorem}

\begin{proof}
In the first stage of our algorithm we start picking sets greedily (i.e. in each step, we pick the set that covers the largest number of previously uncovered elements) until there are no sets that cover at least $t+1$ previously uncovered elements. If during this stage we pick $s$ sets and cover at least $st+tk$ elements then our instance is a \textsc{Yes}-instance. Indeed, we can cover the remaining elements using $\lceil\frac{n-st-tk}{t}\rceil$ sets. In total, such a covering has at most  $\lceil\frac{n-st-tk}{t}\rceil+s = \lceil\frac{n}{t}\rceil-k$ subsets. Intuitively, each set picked during this greedy stage covers at least one additional element. Therefore, if we pick more than $tk$ subsets then our input instance is a \textsc{Yes}-instance. This means that after the greedy stage we either immediately conclude that our input is a \textsc{Yes}-instance or we have used at most $tk$ subsets and covered at most $tk\cdot t+tk=O(k)$ elements, for a fixed $t$. Let us denote the subset of all covered elements by $G$. Note that there is no subset that covers more than $t$ elements from $U\setminus G$. 

If there is a covering of size $\lceil\frac{n}{t}\rceil-k$ then there are $s'\leq |G| \leq tk \cdot t+tk$ subsets that cover $G$ and at least $s't + tk-|G|$ elements of $U\setminus G$. Moreover, if such subsets exist then our input is a \textsc{Yes}-instance.

Hence, it is enough to find $s'$ such subsets. For each $s''\in [|G|]$ we attempt to find subsets $S_1, S_2, \dots, S_{s''}$ such that $G\subset S_1 \cup S_2 \cup \dots \cup S_{s''}$ and $S_1 \cup S_2 \cup \dots \cup S_{s''}$ contains at least $s''t+tk-|G|$ elements from $U\setminus G$. Assume that for some $s''$ such sets exist. Let $H$ be an arbitrary subset of $(S_1 \cup S_2 \cup \dots \cup S_{s''}) \setminus G$ of size $s''t+tk-|G|$. Note that we do not know the set $H$. 
However, we employ the color-coding technique, and construct a $(n,|H|)$-perfect hash family $\mathcal{F}$.
We iterate over all $f \in \mathcal{F}$.
Recall that $|H|=s''t+tk-|G|$.
Now using dynamic programming we can find $H$ in time $2^{O(k)}$ in a standard way.
In order to do that, we consider a new universe $U'$ which contains elements from the set $G$ and elements corresponding to $|H|$ colors corresponding to values assigned by $f$ to $U\setminus G$.
Moreover, if a subset $P$ was a subset that can be used for covering in  \textsc{$t$-Uniform Set Cover} then we replace it with $(P\cap G) \cup \{\text{ all values $f$ assign to elements in } P\cap (U\setminus G)\}$. We replace our \textsc{$t$-Uniform Set Cover} instance with an instance of \textsc{Set Cover} with a universe of size $|G|+|H|\leq |G| + s''t+tk-|G|\leq |G| \cdot t+tk\leq (tk \cdot t+tk)\cdot t+tk=O(k)$. It is easy to see that our original instance is a \textsc{Yes}-instance if and only if the constructed instance of \textsc{Set Cover} admits a covering by at most $s''$ sets (under the assumption that $f$ indeed assigns distinct values to elements of $H$). If $f$ does not assign distinct values then a \textsc{Yes}-instance can be become a \textsc{No}-instance. However, a \textsc{No}-instance cannot become a \textsc{Yes}-instance.
Since $|H|=O(k)$, the overall running time is $2^{O(k)} poly(n)$. 
\end{proof}
As a corollary of the previous theorem we get the following result.
\belowguarantee*
\begin{proof} The result immediately follows from Theorem~\ref{thm:uniformset} with $t=2$.
\end{proof}


\bibliography{literature}

\end{document}